\documentclass[11pt,a4paper]{article}
\usepackage{graphicx}
\usepackage{amsmath}
\usepackage{amsfonts}
\usepackage{amssymb}
\usepackage{amsthm}
\usepackage{algorithm}
\usepackage{algorithmic}
\newtheorem{mydef}{Definition}
\newtheorem{mycoro}{Corollary}
\newtheorem{myprop}{Proposition}
\newtheorem{mytheo}{Theorem}
\newtheorem{mylemm}{Lemma}
\newtheorem{myconj}{Conjecture}
\newtheorem{myexam}{Example}
\author{Andriyan B. Suksmono}
\begin{document}
\title{Probabilistic Construction and Analysis of  Seminormalized Hadamard Matrices}
\maketitle
%
\begin{abstract}
Let $\vec{o}$ be a $4k$-length column vector whose all entries are $1$s, with $k$ a positive integer. Let $V=\{\vec{v}_i\}$ be a set of semi-normalized Hadamard (SH)-vectors, which are $4k$-length vectors whose $2k$ entries are $-1$s and the remaining $2k$ are $1$s. We define a $4k$-order QSH (Quasi SH)-matrix, $\vec{Q}$, as a $4k \times 4k$ matrix where the first column is $\vec{o}$ and the remaining ones are distinct $\vec{v}_i \in V$. When $\vec{Q}$ is orthogonal, it becomes an SH-matrix $\vec{H}$. Therefore, $4k$-order SH-matrices can be built by enumerate all possible $\vec{Q}$ from every combination of $\vec{v}_i$, then evaluate the orthogonality of each one of them. Since such exhaustive method requires a large amount of computing resource, we can employ probabilistic algorithms to construct $\vec{H}$, such as, by Random Vector Selection (RVS) or the Orthogonalization by Simulated Annealing (OSA) algorithms. We demonstrate the constructions of low-order SH-matrices by using these methods. We also analyze some probabilistic aspects of the constructions, including orthogonal probability $p_{\perp}$ between a pair of randomly selected SH-vectors, the existence probability $p_{\vec{H}|\vec{Q}}$ that a randomly generated $\vec{Q}$ is in fact an SH-matrix $\vec{H}$, and address the discrepancy of the distribution between the known number of SH-matrix with expected number derived from the probabilistic analysis.

\end{abstract}

%
\section{Introduction}
A Hadamard matrix is a square matrix whose entries are $1$ or $-1$ and each pairs of its distinct rows (or columns) are orthogonal. These kinds of matrices were first studied by Sylvester \cite{Sylvester1867}. Hadamard realized this kind of matrix when investigating maximal determinant problem \cite{Hadamard1893}. Because of its interesting properties and potential applications, the Hadamard matrix has been extensively studied by many mathematicians and engineers. Formally, we can formulate the matrix by using the following definition.

\begin{mydef}
\label{Def_Hmatrix}
An $m$-order Hadamard matrix $\vec{H}$, where $m$ a positive integer, is an $m \times m$ orthogonal matrix whose entry is either $1$ or $-1$.
\end{mydef} 
We will write the "Hadamard matrix" shortly as the "H-matrix". The orthogonal property in the definition implies 
\begin{equation}
\vec{H}^T \vec{H}= m\vec{I}
\end{equation}
where $\vec{I}$ is an $m$-order, i.e. $m \times m$ identity matrix, and $(\cdot)^T$ denotes the transpose. 

In this paper, when the order of a matrix or a vector is clear from the context of discussion, which in most of cases are $4k$ for a positive integer $k$, it will not be written. Furthermore, we say a square matrix $\vec{A}$ of size $m \times m$ as an $m$-order matrix $\vec{A}$, and an $m$-length or of size $m \times 1$ vector $\vec{v}$ as an $m$-order vector $\vec{v}$. However, when the order needs to be explicitly displayed, it is written as a left-superscript. Therefore, an $m$-order matrix $\vec{H}$ will be written as ${}^m\vec{H}$, and similarly, an $m$-order vector $\vec{v}$ will be written as ${}^m\vec{v}$. 

For conciseness, we adopt shorthand notations for the entries; i.e., the entry $1$ will be written as "+" whereas $-1$ will be written as "-". For examples, H-matrices of order $1$, $2$, and $4$, whose orders are explicitly written, are displayed as follows

\begin{equation}
{}^1\vec{H}=\left(+\right),
{}^2\vec{H}=
\begin{pmatrix}
  + & + \\
  + & -
\end{pmatrix},
{}^4\vec{H}=
\begin{pmatrix}
  + & + & + & +\\
  + & - & + & -\\
  + & + & - & -\\
  + & - & - & +
\end{pmatrix}
\end{equation}

One of the most important properties of the H-matrix is its orthogonality which enables practical application in telecommunications and signal processing. In digital communications, the Hadamard-Walsh codes have been used as a spreading code in CDMA (Code Division Multiple Access) systems \cite{Garg2007,Seberry2005}. The H-matrices are also used to construct an ECC (Error Correcting Code) to protect the messages sent over an unreliable and noisy channel \cite{BoseShrikande1959}. The orthogonal property also make the usage of the Hadamard matrix in signal compression become possible  \cite{Jain1989}. An imaging-related application of the H-matrix is in the 3D spectral imaging described in \cite{Hanley2000}. Further lists and examples of its applications can be found in \cite{HedayatWallis1978} or, more recently, in \cite{Seberry2005}.

An important issue in the H-matrix area is on the construction. Sylvester have used Kronecker product to construct higher order H-matrices from lower order ones. In this method, a $2^k$-order H-matrix ${}^{2^k}\vec{H}$ can be constructed from lower order H-matrices ${}^{2^k-1}\vec{H}$ and ${}^2\vec{H}$ by Kronecker product as follows

\begin{equation}
{}^{2^k}\vec{H}={}^{2^k-1}\vec{H} \otimes {}^2\vec{H}
=
\begin{pmatrix}
  {}^{2^k-1}\vec{H} & {}^{2^k-1}\vec{H}\\
  {}^{2^k-1}\vec{H} & -{}^{2^k-1}\vec{H}
\end{pmatrix}
\end{equation}
Therefore, any $m=2^k$ order H-matrix, where $k$ is a non-negative integer, can be constructed by using the Sylvester method. Consequently, the existence of $2^k$ order H-matrix is guaranteed, which is formulated in the following (Sylvester's) lemma.

\begin{mylemm}
There is a Hadamard matrix of order $2^k$ for every positive integer $k$.
\end{mylemm}

Reordering or exchanging the rows (or columns) of an H-matrix, transposition, and/or multiplying the rows (column) with $-1$ yields another H-matrix. We call the set of H-matrices obtained by such operations as equivalent H-matrices. 

An H-matrix is said to be seminormalized if the entries of either of the first column or the first row are $1$, and it is normalized if the entries in both of the first column and the first row are $1$. In this paper, we refer seminormalized H-matrix as the one whose entries of the first columns are 1. Any H-matrix can be normalized or semi-normalized by exchanging and/or negating rows or column, which yield an equivalent matrix to the original one.

\begin{myexam}
The followings are equivalent operations to a $4$ order H-matrix, ${}^4 \vec{H}$, that yields another $4$ order H-matrices.

\begin{equation}
\label{EQ_H4}
\begin{pmatrix}
  + & - & - & -\\
  - & - & + & -\\
  + & - & + & +\\
  + & + & + & -
\end{pmatrix} 
\to
\begin{pmatrix}
  + & - & - & -\\
  + & + & - & +\\
  + & - & + & +\\
  + & + & + & -
\end{pmatrix}
\to
\begin{pmatrix}
  + & + & + & +\\
  + & - & + & -\\
  + & + & - & -\\
  + & - & - & +
\end{pmatrix}
\nonumber
\end{equation}
These matrices are equivalents because the second matrix is obtained from the first one by negating the $2^{nd}$ row, while the third matrix is obtained from the second one by negating the $2^{nd}$, $3^{rd}$, and $4^{th}$ columns. The second matrix is an SH (Seminormalized Hadamard)-matrix, whereas the last one is an NH (Normalized Hadamard)-matrix.
\end{myexam}

In addition to the Sylvester method, various kinds of construction techniques have been developed, among others are: Paley's method which is based on finite field \cite{Paley1933}, Dade and Goldberg method that is based on permutation group \cite{DadeGoldberg1959}, the Williamson method \cite{Williamson1944}, finite projective plane based method of Bush \cite{Bush1971a,Bush1971b}, orhogonal design by J. Wallis \cite{Wallis1976}, and computer backtracking method by Hall and Knuth \cite{HallKnuth1965}.

Another important issue in the H-matrix is on its existence. It can be proved that if $\vec{H}$ is an $m$-order H-matrix, then $m$ should be a multiple of four, which is formulated in the following proposition.

\begin{myprop}
If $\vec{H}$ is an $m$-order Hadamard matrix, then $m=4k$ for a positive integer $k$.
\end{myprop}
Additionally, it is believed that there is a Hadamard matrix of order $4k$ for every positive integer $k$, which is formulated in the following conjecture.

\begin{myconj}
For every positive integer $k$, there is a Hadamard matrix of order $m=4k$.
\end{myconj}

The Hadamard Matrix Conjecture has been verified for $m\leq428$. The Paley's construction gives H-matrices of order $m=q+1=4k$, where $q$ is any prime power  and $q \equiv 3 \mod 4$. Baumert, Golomb, and Hall found H-matrix of order 92 and 156, and Baumert found H-matrix of order 116 and 232. The largest order known H-matrix for any $4k$-order is 428, which has been discovered by H. Kharaghani and B. Tayfeh-Rezaie \cite{Kharaghani2005}. 

In this paper, we propose to construct the SH-matrix from SH vectors, which are $4k$ length vectors with balanced number of $\{ 1, -1\}$ entries as candidates of the column vectors of the SH-matrix. The rest of the paper is organized as follows. In Section 2, we formulate a concept of the SH (Seminormalized Hadamard) Vectors and their properties. Important results of this chapter is the orthogonality relationship among the SH vector, and quantification of orthogonal probability of randomly selected SH vector. Section 3 describe the construction of SH matrix from the SH vectors. We present three construction method, i.e., exhaustive search RVS (Random Vector Selection), and OSA (Orthogonalization by Simulated Annealing). In Section 4, we give the probabilistic analysis that leads into the formulation of the orthogonal probability of a pair of randomly generated SH-vectors and propose existence probability of a $4k$-order SH matrix.
%
\section{Seminormalized Hadamard Vectors and Their Properties}

In this section, we develop the concept of SH (Seminormalized Hadamard) vectors $\vec{v}_i \in V$ and determine the number of distinct $4k$-order SH vectors $|V|=N_V$. Then, we identify the orthogonality relationship among the SH vectors and determine the number $N_O$ of vectors in $V$ that are orthogonal to a given $\vec{v}_i \in V$. We will also show that the orthogonality relationship among the SH vectors in $V$ can be represented by an $N_O-$regular $N_V-$order graph.

We will be dealing with vector whose entries are $1$ and $-1$, mostly with even length and balanced number of $1$ and $-1$. In this section, we build such vectors from their basic building blocks of unity vectors, its negated values, concatenations, and rearrangements. We start with the following definition.

\begin{mydef}
\label{DEF_kUnitiyVec}
A $k$-order unity vector, where $k$ a positive integer, is a vector whose (all of its) entries are $+1$
\begin{equation}
 {}^k\vec{1}\equiv \left(
\underbrace{ 
 \begin{array}{cccc}
  + & + & \cdots & +
 \end{array}
 }_{k-number}
 \right)^T
\end{equation}
And its negated pair is
\begin{equation}
 -{}^k\vec{1} \equiv \left(
\underbrace{ 
 \begin{array}{cccc}
  - & - & \cdots & -
 \end{array}
 }_{k-number}
 \right)^T
\end{equation}.

\end{mydef}
From the $k$-order unity vectors $ \{ \pm {}^k\vec{1} \}$, we can construct the $2k$-order vectors by concatenations as follows,

\begin{equation}
\left(\begin{array}{cc}  {}^k\vec{1}^T & {}^k\vec{1}^T \end{array} \right)^T = \left(
\underbrace{ 
 \begin{array}{cccc}
  + & + & \cdots & +
 \end{array}
 }_{k-number}
\underbrace{ 
 \begin{array}{cccc}
  + & + & \cdots & +
 \end{array}
 }_{k-number}
 \right)^T \equiv {}^{2k}\vec{1}
\end{equation}
%
\begin{equation}
\left(\begin{array}{cc}  -{}^k\vec{1}^T & -{}^k\vec{1}^T \end{array} \right)^T = \left(
\underbrace{ 
 \begin{array}{cccc}
  - & - & \cdots & -
 \end{array}
 }_{k-number}
\underbrace{ 
 \begin{array}{cccc}
  - & - & \cdots & -
 \end{array}
 }_{k-number}
 \right)^T \equiv -{}^{2k}\vec{1}
\end{equation}
\begin{equation}
\label{EQ_OHH_vector}
\left(\begin{array}{cc}  -{}^k\vec{1}^T & {}^k\vec{1}^T \end{array} \right)^T = \left(
\underbrace{ 
 \begin{array}{cccc}
  - & - & \cdots & -
 \end{array}
 }_{k-number}
\underbrace{ 
 \begin{array}{cccc}
  + & + & \cdots & +
 \end{array}
 }_{k-number}
 \right)^T \equiv \vec{s}
\end{equation}

\begin{equation}
\left(\begin{array}{cc}  {}^k\vec{1}^T & -{}^k\vec{1}^T \end{array} \right)^T = \left(
\underbrace{ 
 \begin{array}{cccc}
  + & + & \cdots & +
 \end{array}
 }_{k-number}
\underbrace{ 
 \begin{array}{cccc}
  - & - & \cdots & -
 \end{array}
 }_{k-number}
 \right)^T \equiv -\vec{s}
\end{equation}

In particular, we will be interested to the vector $\vec{s}$  given by Equation \ref{EQ_OHH_vector}, which is further formalized in the following definition.

\begin{mydef}
\label{DEF_OHH_vector}
A $2k$-order OHH (Ordered Half-length Hadamard) vector $\vec{s}$, where $k$ is a positive integer, is a vector constructed by concatenating $-{}^k \vec{1}$ with ${}^k \vec{1}$, i.e., it is the vector $\vec{s} \equiv \left(\begin{array}{cc}  -{}^k\vec{1}^T & {}^k\vec{1}^T \end{array} \right)^T$.
\end{mydef}

By inspection of Eq. \ref{EQ_OHH_vector}, we realize that the OHH vector has balanced number of 1 and -1. Accordingly, its inner product with $2k$-length unity vector ${}^{2k} \vec{1}$ will be zero, which is formulated in the following lemma.

\begin{mylemm}
\label{LEMM_OHH_ortho}
The $2k$ order OHH vector $\vec{s}$ is orthogonal to $2k$-order unity vector, i.e., $\left< \vec{s},{}^{2k} \vec{1} \right>=0$.
\end{mylemm}

\begin{proof}
By expanding the inner product, we obtain the followings:

\begin{equation}
\left< \vec{s},{}^{2k} \vec{1} \right> = \Sigma_{j=1}^{2k} \vec{s}[j] \cdot {}^{2k}\vec{1}[j] = \Sigma_{j=1}^{2k} \vec{s}[j]\cdot 1 = \Sigma_{j=1}^{2k} \vec{s}[j]
\nonumber
\end{equation}
Then, by Definition  \ref{DEF_OHH_vector} of the OHH vector, we arrive to the following result

\begin{equation}
 \Sigma_{j=1}^{2k} \vec{s}[j]= \underbrace{-1 -1 ... -1}_{k-number} + \underbrace{ 1 +1 ... +1}_{k-number} = 0. 
\nonumber
\end{equation}

\end{proof}

In the next formulation, we will be working with rearranged vectors, which is achieved by permutation of a vector's entries. First, lets define a set $\Sigma$ of $n$-object permutations as follows.

\begin{mydef}
\label{DEF_permSet}
The set of all permutations of $n$-objects will be called the set of $n$-order permutations $\Sigma=\{ \sigma_i\}$.
\end{mydef}
Furthermore, the permutation operation to the  entries of a vector $\vec{v}$ can be defined as follows.

\begin{mydef}
\label{DEF_permVec}
Consider an $n$-order permutation $\sigma_i \in \Sigma$. Permutation of an $n$-order vector $\vec{v}=\vec{v}[j]$ by $\sigma_i$ is a rearrangement of the vector’s entries according to $\sigma_i$, which is denoted by $\sigma_i \left(\vec{v}\right)$. 
\end{mydef}
To clarify the concept of vector permutation, consider the following example.

\begin{myexam}
\label{EXAM_4SHV}
Consider a $4$-order vector  $\vec{v}_1=\left(\begin{array}{cccc} 1& 2&3&4 \end{array} \right)^T$. A permutation by $\sigma_i=\left(\begin{array}{cccc} 1&2&3&4 \\ 2&3&1&4 \end{array} \right)$ yields $\sigma_i \left(\vec{v}_1\right)=\left(\begin{array}{cccc} 3&1&2&4 \end{array} \right)^T$. The same permutation $\sigma_i$ to a vector $\vec{v}_2=\left(\begin{array}{cccc} 1&-1&-1&1 \end{array} \right)^T$ yields $\sigma_i \left(\vec{v}_2\right)=\left(\begin{array}{cccc} -1&1&-1&1 \end{array} \right)^T$.
\end{myexam}
It is easy to show that an $n$-order unity vector will not be changed by a permutation. We formulate this fact into the following lemma.

\begin{mylemm}
\label{LEMM_invariantPerm}
The $n$-order unity vector ${}^n\vec{1}$ is invariant under $n$-order permutation $\sigma_i$, i.e., $\sigma_i \left( {}^n\vec{1}\right)={}^n\vec{1}$.
\end{mylemm}

\begin{proof}
Since the entries of the $n$-order unity vector ${}^n\vec{1}$ are identical, which are 1s, rearrangement by $n$-order permutation $\sigma_i$ will not change the vector, i.e.,

\begin{equation}
\sigma_i \left( {}^n\vec{1}\right)= \sigma_i \left( \left( \begin{array}{cccc} 1&1&...&1 \end{array} \right)^T \right)= \left( \begin{array}{cccc} 1&1&...&1 \end{array} \right)^T = {}^n\vec{1}
\nonumber
\end{equation}
\end{proof}

Furthermore, applying identical permutation will also not change the value of their inner product. This result will be useful to proof some lemmas or theorems in this paper. Therefore, we formulate it in the following lemma.  

\begin{mylemm}
\label{LEMM_invperm_innerprod}
Let $\sigma_i$ be an $n$-order permutation and $\vec{v}_p, \vec{v}_q $ be $n$-order vectors. Applying $\sigma_i$ identically to both of these vectors does not change their inner product, i.e., 
\begin{equation}
 \left<\sigma_i(\vec{v}_p),\sigma_i(\vec{v}_q)\right> =
 \left<\vec{v}_p,\vec{v_q} \right>
\nonumber
\end{equation}
\end{mylemm}
This lemma is a direct consequence of commutative property of the addition, i.e., rearrangement by $\sigma_i$ does not change the results of the terms in addition of $\vec{v}_p[j]\vec{v_q}[j]$ in the calculation of the inner products.

\begin{proof}
Consider an $n$-order vector $\vec{w}$ whose entries are the product of entry-wise multiplication of $\vec{v}_p$  and $\vec{v}_q$, i.e., $\vec{w}[j]=\vec{v}_p[j] \vec{v}_q[j]$. Then,

\begin{equation}
\left<\vec{v}_p,\vec{v_q} \right>=\Sigma_{j=1}^n \vec{v}_p[j] \vec{v}_q[j] =\Sigma_{j=1}^n \vec{w}[j] 
\nonumber
\end{equation} 
Since the permutation $\sigma_i$ changes the indices from $j$ to a new one, let say $r$, then

\begin{equation}
 \left<\sigma_i(\vec{v}_p),\sigma_i(\vec{v}_q)\right> =
\Sigma_{j=1}^n \sigma_i(\vec{v}_p[j)]\sigma_i(\vec{v_q}[j])= \Sigma_{r=1}^n \vec{v_p}[r] \vec{v_q}[r]
\nonumber
\end{equation}

\begin{equation}
 =\Sigma_{r=1}^n \vec{w}[r]=\Sigma_{j=1}^n \vec{w}[j]=\left<\vec{v}_p,\vec{v_q} \right>.
\nonumber
\end{equation}

\end{proof}

Now, we are ready to construct various kinds of H-vectors. Let's start with HSH (Half-length Seminormalized Hadamard) vectors, which are any $2k$ length vector with balanced number of 1 and -1.

\begin{mydef}
\label{DEF_HSH}
A $2k$-order HSH (Half-length Seminormalized Hadamard) vector $\vec{f}_i \in F$, where $k$ is a positive integer, is a vector which is obtained by $2k$-order permutation $\sigma_i \in \Sigma$ of the OHH vector $\vec{s}$, i.e., it is the vector given by $\vec{f}_j=\sigma_i \left( \vec{s} \right)$.
\end{mydef}

Although $\sigma_i$ is bijective, $2k$ number of entries in $\vec{s}$ consist of only two kinds of objects, i.e., $\{-1,1\}$. Accordingly, different permutations might yields identical HSH vectors, which implies that the cardinality of $F$ cannot be directly determined by the permutation, due to the redundancy of the entries; i.e. $|F|<(2k)!$. The following lemma gives the correct number of $2k$-order HSH vectors.

\begin{mylemm}
\label{LEMM_Num_HSH}
There are $N_F \equiv |F|=C(2k,k)$ number of distinct HSH vectors of order $2k$.
\end{mylemm}
To proof the lemma, we employ a simple counting argument, i.e., the combination of $k$ objects from the set of $2k$ objects.

\begin{proof}
In a $2k$-order HSH vector, there are $k$-number of +1 and $k$ number of -1. Consider them as two different kinds of objects, which will be arranged into $2k$ places. By the counting principle, there are $C(2k,k)$ ways to arrange the first kind of the objects, i.e. $1$, into $2k$ places and fill the rest with the second ones, i.e., $-1$.
\end{proof}

Since the HSH has a balanced number of 1 and -1, we can show that every HSH vector is orthogonal to the $2k$-length unity vector. We formulate this fact into the following lemma.

\begin{mylemm}
\label{LEMM_HSH_ortho}
Every $2k$-order HSH vectors $\vec{f}_i \in F$ is orthogonal to the $2k$ length unity vector ${}^{2k} \vec{1}$, i.e., $\left< \vec{f}_j, {}^{2k}\vec{1}\right>=0$.
\end{mylemm}
We employ previous results to proof the lemma.

\begin{proof}
Consider an HSH vector $\vec{f}_i \in F$. Based on Definition \ref{DEF_HSH} and by employing Lemmas \ref{LEMM_invariantPerm}, \ref{LEMM_invperm_innerprod}, and \ref{LEMM_OHH_ortho} subsequently, we obtain

\begin{equation}
\left< \vec{f}_j,{}^{2k}\vec{1} \right> = 
\left< \sigma_i(\vec{s}), {}^{2k}\vec{1} \right> = 
\left< \sigma_i(\vec{s}), \sigma_i({}^{2k}\vec{1}) \right>
=\left< \vec{s}, {}^{2k}\vec{1} \right>=0.
\nonumber
\end{equation}
\end{proof}

It is trivial to show that $\left< \vec{f}_j,-{}^{2k}\vec{1} \right> = 0$ is also hold. Moreover, any permuted HSH vector is also orthogonal to the $2k$-order unity vector, which is formulated in the following lemma.

\begin{mylemm}
\label{LEMM_permHSH_ortho}
Permuted $2k$-order HSH vectors are orthogonal to the $2k$ order unity vector, i.e., $ \left< \sigma_i(\vec{f}_j),{}^{2k} \vec{1}\right> = 0$.
\end{mylemm}

\begin{proof}
Consider a $2k$-order HSH vector $\vec{f}_j \in F$, a $2k$-order permutation $\sigma_i \in \Sigma$, and the $2k$ order unity vector ${}^{2k}\vec{1}$. Then, by employing Lemmas \ref{LEMM_invariantPerm}, \ref{LEMM_invperm_innerprod}, and \ref{LEMM_HSH_ortho} subsequently, we obtain
\begin{equation}
\left<\sigma_i(\vec{f}_j), {}^{2k} \vec{1} \right> = \left<\sigma_i(\vec{f}_j), \sigma_i({}^{2k} \vec{1}) \right>  =
 \left<\vec{f}_j,{}^{2k} \vec{1} \right>=0 
 \nonumber
\end{equation}

\end{proof}

In the next stage, we construct various H-vectors of length $4k$. We start with the definition of a $4k$ length unity vector ${}^{4k}\vec{1}$ and an ordered seminormalized Hadamard vector as follows.

\begin{mydef}
\label{DEF_4kunity}
A $4k$-order unity vector $\vec{o} \equiv {}^{4k}\vec{1}$, with $k$ a positive integer, is a vector of length $4k$ whose all entries are 1, i.e., it is a vector of the following form

\begin{equation}
 \vec{o} \equiv {}^{4k}\vec{1} = \left(
\underbrace{ 
 \begin{array}{cccc}
  + & + & \cdots & +
 \end{array} }_{4k-number} \right)^T
\end{equation}

\end{mydef}

\begin{mydef}
\label{DEF_OSH}
An OSH (Ordered Seminormalized Hadamard) vector $\vec{t}$ is a $4k$-length vector of the following form
\begin{equation}
\vec{t} \equiv \left( 
\begin{array}{cccc}
  -{}^k\vec{1}^T & -{}^k\vec{1}^T & {}^k\vec{1}^T & {}^k\vec{1}^T \end{array}
\right)^T
\end{equation}
where $k$ is a positive integer.
\end{mydef}

Since the number of 1 and -1 in the OSH vector is balanced, it will eventually orthogonal to the $4k$-order unity vector $\vec{o}$. We formulate this fact into the following lemma.

\begin{mylemm}
\label{LEMM_OSH_ortho4kunity}
The $4k$-order OSH vector is orthogonal to $4k$-order unity vector $\vec{o}$, i.e., $\left<\vec{t}, \vec{o} \right>=0$.
\end{mylemm}

\begin{proof}
By Definition \ref{DEF_OSH} and expansion of the inner product, we obtain

\begin{equation}
\left<\vec{t}, \vec{o} \right> = \Sigma_{j=1}^{4k} \vec{t}[j] \cdot {}^{4k}\vec{1}[j] =  \Sigma_{j=1}^{4k} \vec{t}[j] \cdot 1 = \Sigma_{j=1}^{4k} \vec{t}[j] 
\nonumber
\end{equation}

\begin{equation}
= \underbrace{ 
   -1-1 \cdots -1}_{2k-number}+ \underbrace{ 
   1+1 \cdots +1}_{2k-number} = 0
\nonumber
\end{equation}

\end{proof}

Finally, we can define the Seminormalized Hadamard vector similarly to the HSH vector as follows.

\begin{mydef}
\label{DEF_SHvector}
A $4k$-order SH (Seminormalized Hadamard) vector $\vec{v}_j \in V$ is a vector that is obtained by $4k$-order permutation $\sigma_i \in \Sigma$ of the $4k$-order OSH vector $\vec{t}$, i.e, $\vec{v}_j =\sigma_i(\vec{t})$.
\end{mydef}

Since the identity permutation $\sigma_0 \in \Sigma$, by Definition \ref{DEF_SHvector}, it is trivial that the OSH vectors $\vec{t}$ is also an SH vectors. Additionally, since the permuted objects are redundant -1 and 1, different permutation may yields identical SH vectors, therefore $|V|<(4k)! $. The correct number of the $4k$ order SH vector is given by the following lemma.

\begin{mylemm}
\label{LEMM_SH_number}
There are $N_V=C(4k,2k)$ number of distinct SH vectors of order $4k$.
\end{mylemm}

\begin{proof}
Consider again two kinds of objects, i.e., $2k$ number of 1, and $2k$ number of -1, which are arranged in $4k$ places. By counting principle, there are $C(4k,2k)$ distinct ways to arrange $2k$  number of the first object (1) into $4k$ places and fill the remaining with the second one (-1).
\end{proof}

We also have orthogonality property of SH vector to the $4k$ order unity vector, given by the folloing lemma.

\begin{mylemm}
\label{LEMM_SH_ortho}
The $4k$ order SH vectors $\vec{v}_j \in V$ are orthogonal to $4k$-order unity vector $\vec{o}$.
\end{mylemm}

\begin{proof}
Consider an SH vector $\vec{v}_j \in V$. Then, by Definition \ref{DEF_SHvector}, and applying Lemmas \ref{LEMM_invariantPerm} and \ref{LEMM_OSH_ortho4kunity} subsequently, we get

\begin{equation}
\left< \vec{v}_j,\vec{o} \right> = \left< \sigma_i(\vec{t}),\vec{o} \right> = \left< \sigma_i(\vec{s}), \sigma_i(\vec{o}) \right> = \left< \vec{t}, \vec{o} \right>=0
\nonumber
\end{equation}

\end{proof}

Lemma \ref{LEMM_HSH_ortho} states that an HSH vector $\vec{f}_j$ is orthogonal to $2k$ order unity vector ${}^{2k}\vec{1}$; and accordingly to the negated $2k$-order unity vector $-{}^{2k}\vec{1}$ as well. Consider two HSH vectors $\vec{f}_1$ and $\vec{f}_2$. The calculation of the inner product of concatenated two HSH vectors $\left( \begin{array}{cc}
\vec{f}_1 & \vec{f}_2
\end{array} \right)^T$
with the OSH vector can be conducted block-wise. By applying orthogonal property of HSH vector to $2k$ order unity vector, we obtain

\begin{equation}
\left< \left( \begin{array}{cc} \vec{f}_1^T & \vec{f}_2^T
\end{array} \right)^T, \left( \begin{array}{cc} -{}^{2k}\vec{1}^T & {}^{2k}\vec{1}^T
\end{array} \right)^T \right> = \left<\vec{f}_1,-{}^{2k}\vec{1}\right> + \left<\vec{f}_1,{}^{2k}\vec{1}\right> 
\nonumber
\end{equation}
\begin{equation}
 = 0+0 = 0
\nonumber
\end{equation}
But, the vector $\left( \begin{array}{cc} -{}^{2k}\vec{1}^T & {}^{2k}\vec{1}^T
\end{array} \right)^T = \left( \begin{array}{cccc} -{}^{k}\vec{1}^T&-{}^{k}\vec{1}^T&{}^{k}\vec{1}^T&{}^{k}\vec{1}^T 
\end{array} \right)^T $ is in fact the OSH vector $\vec{t}$. Accordingly, two concatenated HSH vectors are orthogonal to an OSH vector, i.e.,

\begin{equation}
\left< \left( \begin{array}{cc} \vec{f}_1^T & \vec{f}_2^T
\end{array} \right)^T, \left( \begin{array}{cc} -{}^{2k}\vec{1}^T&{}^{2k}\vec{1}^T
\end{array} \right)^T \right> = 
\left< \left( \begin{array}{cc} \vec{f}_1^T & \vec{f}_2^T
\end{array} \right)^T,\vec{t}\right> = 0
\nonumber
\end{equation}
Furthermore, concatenation of permuted $\vec{f}_1$ by $\sigma_p$ and $\vec{f}_2$ by $\sigma_q$ yields vectors that are orthogonal to the OSH vector, since by block vector operation and Lemma \ref{LEMM_permHSH_ortho}, we can express

\begin{equation}
\left< \left( \begin{array}{cc} \sigma_p(\vec{f}_1)^T & \sigma_q(\vec{f}_2)^T \end{array} \right)^T, \left( \begin{array}{cc} -{}^{2k}\vec{1}^T&{}^{2k}\vec{1}^T
\end{array} \right)^T \right>=0
\nonumber
\end{equation}
which motivate us to define the PSH (Partitioned SH) vectors as follows.

\begin{mydef}
\label{DEF_PSH}
Let $\sigma_p(\vec{f}_1)$ and $\sigma_p(\vec{f}_2)$ be permuted HSH vectors. A $4k$-order vector $\vec{w}_i \in W$ constructed by concatenation of these two vectors, $\vec{w}_i=\left( \begin{array}{cc} \sigma_p(\vec{f}_1)^T &\sigma_p(\vec{f}_2)^T \end{array} \right)^T $ is called a partitioned SH (PSH) vector.
\end{mydef}

In a PSH vector, the number of -1 and 1 is maintained to be in balanced at each of the block/ partition, i.e., each of half part of the $4k$-length vector that consists of $2k$ entries in the left, and another $2k$ entries in right parts. This balancing makes PSH vectors orthogonal to the OSH vector, which is formulated in the following lemma.

\begin{mylemm}
\label{LEMM_PSH_ortho}
All of the $4k$-order PSH vectors $\vec{w}_i \in W$ are orthogonal to the $4k$-order OSH vector $\vec{t}$, i.e., $\left< \vec{w}_i, \vec{t} \right>=0$.
\end{mylemm}

\begin{proof}
Based on Definition \ref{DEF_PSH}, property of block vector operation, and Lemmas \ref{LEMM_invariantPerm},  \ref{LEMM_invperm_innerprod}, and \ref{LEMM_HSH_ortho} subsequently, we obtain the following results

\begin{equation}
\left< \vec{w}_i,\vec{t} \right> =
\left< \left( \begin{array}{cc} \sigma_p(\vec{f}_1)^T & \sigma_q(\vec{f}_2)^T \end{array} \right)^T, \left( \begin{array}{cc} -{}^{2k}\vec{1}^T&{}^{2k}\vec{1}^T
\end{array} \right)^T \right>
\nonumber
\end{equation}
\begin{equation}
=\left<\sigma_p(\vec{f}_1), -{}^{2k}\vec{1}\right>+\left<\sigma_q(\vec{f}_2),{}^{2k}\vec{1}\right>
\nonumber
\end{equation}
\begin{equation}
=\left<\sigma_p(\vec{f}_1), \sigma_p(-{}^{2k}\vec{1})\right>+\left<\sigma_q(\vec{f}_2),\sigma_q({}^{2k}\vec{1}) \right>
\nonumber
\end{equation}

\begin{equation}
=\left<\vec{f}_1,-{}^{2k}\vec{1}\right> +\left< \vec{f}_2,{}^{2k}\vec{1} \right> = 0 + 0 = 0
\nonumber
\end{equation}

\end{proof}

The number of $4k$-order PSH vector is useful in our forthcoming calculations. We formulate the number in the following lemma.

\begin{mylemm}
\label{LEMM_Num_PSH}
The cardinality of the set of all $4k$-order PSH vector is $|W|\equiv N_W=C(2k,k)^2$.
\end{mylemm}

\begin{proof}
Based on Lemma \ref{LEMM_Num_HSH}, the number of (distinct) vectors obtained from either the left- or the right-part of $\vec{w}_i$ is $C(2k,k)$. By product rule, we obtain
$ |W| \equiv N_W= C(2k,k)\cdot C(2k,k)= C(2k,k)^2$.
\end{proof}

Now, we arrive to an important result in this Section, i.e., vectors that are orthogonal to a particular SH vector $v_j \in V$, which is given by the following theorem.

\begin{mytheo}
\label{THEO_SH_ortho}
Let $\vec{v}_j \in V$ be a $4k$-order SH vector and suppose that $\sigma_i \in \Sigma$ is a $4k$-order permutation that transform the $4k$-order OSH vector $\vec{t}$ to the SH vector $\vec{v}_j$, i.e., $\sigma_i(\vec{t}) =\vec{v}_j$. Then, all of the vectors obtained by the permutation $\sigma_i$ to PSH vectors $\vec{w}_n \in W$, are orthogonal to $\vec{v}_j$, i.e., $\left< \sigma_i(\vec{w}_n),\vec{v}_j \right>=0$.
\end{mytheo} 

\begin{proof}
Consider the ($4k$ order) OSH vector $\vec{t}$, an SH vector $\vec{v}_j \in V$, PSH vectors $\vec{w}_n \in W$, and permutation $\sigma_i \in \Sigma$, so that $\sigma_i(\vec{t})= \vec{v}_j$, and therefore $\vec{t}=\sigma_i^{-1}(\vec{v}_j)$. Then, based on Lemmas \ref{LEMM_PSH_ortho} and \ref{LEMM_invperm_innerprod}, since $\sigma_i \cdot \sigma_i^{-1}=\sigma_0$, we get

\begin{equation}
\left< \vec{w}_n,\vec{t} \right> = 0 
\nonumber
\end{equation} 
\begin{equation}
\iff
\left<\vec{w}_n, \sigma_i^{-1}(\vec{v}_j) \right> = 0 
\nonumber
\end{equation} 
\begin{equation}
\iff
\left<\sigma_i(\vec{w}_n), \sigma_i(\sigma_i^{-1}(\vec{v}_j)) \right> = 0 
\nonumber
\end{equation} 
\begin{equation}
\iff
\left<\sigma_i(\vec{w}_n), \sigma_0(\vec{v}_j) \right> = 0 
\nonumber
\end{equation} 
\begin{equation}
\iff
\left<\sigma_i(\vec{w}_n), \vec{v}_j \right> = 0 
\nonumber
\end{equation} 
\end{proof}
In practice, since the known vectors are $\vec{v}_i$ and $\vec{t}$ (which is defined), to obtain the set of orthogonal vectors to $\vec{v}_j$, we will first compute $\sigma_i^{-1}$, then calculate $\sigma_i$, and at the final step, we apply it to $W$. 

Based on the Theorem \ref{THEO_SH_ortho}, we obtain the number of vectors that are orthogonal to $\vec{v}_j \in V$ as follows.

\begin{mycoro}
\label{CORO_Num_orthov}
Every $4k$-order SH vector $\vec{v}_i \in V$ is orthogonal to $N_O=C(2k,k)^2$ number of another $4k$-order SH vector in $V$.
\end{mycoro}

\begin{proof}
Recall from Theorem \ref{THEO_SH_ortho} that each $\vec{v}_j \in V$ is orthogonal to $\sigma_i(\vec{w}_n)$, for a particular permutation $\sigma_i$, where $\vec{w}_n \in W \subset V$ are PSH vectors. Since permutation is bijective, it is clear that  $\sigma_i(\vec{w}_n)$ is also an SH vector, and therefore $\sigma_i(W) \subset V$. Furthermore, the bijectivity of the permutation implies that $|\sigma_i(W)|=|W|$. Accordingly, by Lemma \ref{LEMM_Num_PSH}, we obtain $N_O=N_W=C(2k,k)^2$.
\end{proof}

Calculation of orthogonal probability of a pair of randomly selected SH vectors and the probability of existence of H-matrix can be conducted by using orthogonality graph $\Gamma=(V,E)$. This graph represents orthogonality relationship between a pair of SH vectors $\vec{v}_i,\vec{v}_j \in V$, for all SH vectors in $V$. The vertices $V$ of the graph represents SH vectors $\vec{v}_i$, whereas the edges $e_{ij} \in E$ between two vertices $\vec{v}_i,\vec{v}_j$ states that they are orthogonal to each other, i.e., $\vec{v}_i \perp \vec{v}_j$. For a given set of $4k$ order SH vectors, we can make the following statement.

\begin{myprop}
\label{PROP_SHGraph}
The orthogonality relationship of $4k$-order SH vectors can be represented by an $N_V$-order $N_O$-regular (undirected) graph, where $N_V=|V|$ is the number of all SH vectors $\vec{v}_i \in V$, and $N_O$ is the number of all $4k$-order PSH vectors $\vec{w}_j \in W$, i.e., $N_O=|W|$.
\end{myprop}

\begin{proof}
By representing the SH vectors $\vec{v}_i, \vec{v}_j \in V$ as vertices of a graph $\Gamma=(V,E)$ , where an edge between $\vec{v}_i$ and $\vec{v}_j$ represents $\vec{v}_i \perp \vec{v}_j$, Corollary \ref{CORO_Num_orthov} implies that the degree of the vertices are identicals, whose value is $N_O$; therefore $\Gamma$ is $N_O$-regular. Since there are $N_V$ vertices in $\Gamma$, it is an $N_V$-order (undirected) graph.
\end{proof}

\begin{myexam}
Consider a $k=1$ or $4$ order SH vector. Based on the previous lemmas, we have the number of SH vector $N_V=C(4,2)=6$, and each of the vector will be orthogonal to another $N_O=C(2,1)^2=4$ SH vectors. The list of the SH vectors are: $\vec{v}_1=\left( \begin{array}{cccc} +&+&-&-\end{array}\right)^T$, $\vec{v}_2=\left( \begin{array}{cccc} -&-&+&+\end{array}\right)^T$, $\vec{v}_3=\left( \begin{array}{cccc} +&-&+&-\end{array}\right)^T$, $\vec{v}_4=\left( \begin{array}{cccc} +&-&-&+\end{array}\right)^T$, $\vec{v}_5=\left( \begin{array}{cccc} -&+&+&-\end{array}\right)^T$, and $\vec{v}_6=\left( \begin{array}{cccc} -&+&-&+\end{array}\right)^T$. Based on Proposition \ref{PROP_SHGraph}, the orthogonal relationship of the 4-order SH vectors can be represented by a 6-order 4-regular graph, which is shown in Figure \ref{fig:FIG_SH4}.

\begin{figure}
 \centering
 \includegraphics[scale=0.50]{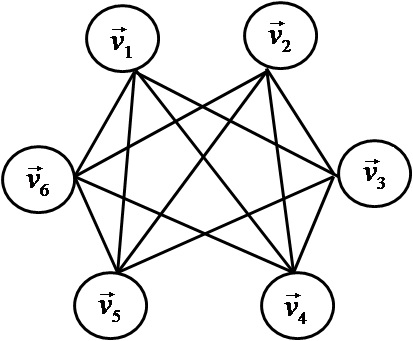}
 \caption{A 6-order 4-regular graph representing orthogonality relationship of 4-order SH vectors}
 \label{fig:FIG_SH4}
\end{figure}

\end{myexam}

\textbf{Remark}. Graph representation of a H-matrix has been described by Ito in \cite{Ito1985}. In his paper, the graph is denoted by $\Delta(m)$, with $m=4k$, for a positive integer $k$.  Vertices of the graphs are $(1,-1)$ vectors of length $4k$ consisting of even number of $1$s, whose adjacency relation consists in orthogonality \cite{Alvarez}.

%
\section{Construction of SH Matrices from SH Vectors}
This section describes construction of SH matrices from the SH vectors formulated in the previous section. A straightforward way to construct the matrix is by listing up all possible combinations of $(4k-1)$ SH vectors $\vec{v}_i \in V$ as the column vectors, along with the unity vector $\vec{o}$, to form $4k$ order candidate matrices $\{\vec{Q}_i\}$; which is called QSH (Quasi SH) matrix. Then, we perform orthogonality test to each of the matrices of the corresponding order. However, such an exhaustive searching needs a large amount of computational resource, since the list of the candidates grows exponentially with the increase of the matrix's order $4k$. 

Alternatively, we can employs a method, which in principle is constructing or finding a single SH matrix of the corresponding order, instead of finding all of them suggested in the exhaustive search. The first method addressing this issue is RVS (Random Vector Selection). The RVS subsequently construct the matrix, column by-column, by generating a random SH vector while maintaining the orthogonality of the matrix in each of the stages. The second method is OSA (Orthogonality by Simulated Annealing), in which, a random QSH matrix is firstly generated, then we flip or exchange a randomly selected pair of $\{1,-1\}$ entries, at  a randomly selected column and checked a deviation error from a $4k$ order SH matrix. Obviously, the RVS and OSA methods are probabilistic in nature.

We evaluate the capability of each method by constructing a (low order) SH matrix which cannot be constructed by the Sylvester's method, such as $k=12$. First, consider the following definition of the candidate matrix $\vec{Q}$.

\begin{mydef}
A $4k$-order QSH (Quasi Seminormalized Hadamard)-matrix $\vec{Q}$, where $k$ a is positive integer, is a $4k\times 4k$ matrix whose first column is (the $4k$-order unity vector) $\vec{o}$ and the rest ($4k-1)$ columns are (distinct) $4k$-order SH vectors.
\end{mydef}
Note that a $4k$-order QSH matrix $\vec{Q}$ is not necessarily orthogonal. When it is orthogonal, then it becomes a $4k$-order SH matrix  $\vec{H}$. 


To enumerate all of the QSH matrix, we define the set of $4k$-order QSH matrix as follows,

\begin{mydef}
The set of $4k$-order QSH matrices $\Theta=\{ \vec{Q}_1, \vec{Q}_2, ..., \vec{Q}_{N_Q}\}$ contains all possible $4k$-order QSH matrix, which is built from $4k$-order SH vectors $\vec{v}_i \in V$ and $4k$-order unity vector $\vec{o}$.
\end{mydef}
By Lemma \ref{LEMM_SH_number}, we know the number $N_V$ of $4k$ order SH vectors. Since each of the QSH matrix needs $(4k-1)$ SH vectors, the number of QSH matrix can be calculated, which is given by the following Lemma.

\begin{mylemm}
\label{LEMM_NumQSH}
The number of $4k$ order QSH matrix is $N_{\vec{Q}}= \frac{N_V!}{\left( N_V-(4k-1)\right)!}$.
\end{mylemm}

\begin{proof}
Recall that the first column of $\vec{Q}$ is a $4k$ order unity matrix $\vec{o}$ and the remaining ones are $(4k-1)$ number of distinct $\vec{v}_i \in V$, with $|V|=N_V$. Then, by basic counting principle, there are $N_V$ ways to fill the $2^{nd}$ column,  $(N_V -1)$ ways to fill the $3^{rd}$, $\cdots$, $\left( N_V-(4k-2) \right)$ to fill the $4k^{th}$ column. Therefore, the total number of QSH matrix is
\begin{equation}
N_{\vec{Q}}=N_V(N_V-1)(N_V-2) \cdots (N_V-(4k-2))= \frac{N_V!}{\left( N_V-(4k-1)\right)!}
\nonumber
\end{equation}
\end{proof}

Note that the number of QSH matrix given by Lemma \ref{LEMM_NumQSH} related to non-unique SH matrices, since a same subset of $(4k-1)$  SH vectors have been used repeatedly in many of the QSH matrix counted by the lemma. The number of unique QSH matrix $N_{\vec{Q},U}$ should consider only non-redundant combination of the subset, which is given by the following lemma,

\begin{mylemm}
The number of unique QSH matrices of order $4k$ is $N_{\vec{Q},U}=C(N_V,4k-1)$.
\end{mylemm}

\begin{proof}
By counting argument, given $N_V$ number of SH matrix to be arranged in $(4k-1)$ column of the QSH matrix, the number of distinct QSH matrices is given by combination $C(N_V,4k-1)$.
\end{proof}

\subsection{Exhaustive Search}
In principle, this method evaluates all possible combinations of $(4k-1)$ SH vectors from $N_v$ number of $v_i \in V$ that satisfy the orthogonality condition. This method starts with the construction of the set of SH vectors $V$. Then, all possible QSH-matrices $\vec{Q}$ are constructed by selecting any $(4k-1)$ combination of $N_V$ number of SH vectors. At the final stage, we evaluate the orthogonality of each of the QSH matrix. The following Algorithm \ref{ALG_xhstvSearch} illustrate a method to generate $V$, which is followed by construction of QSH-matrix, and the evaluation of their orthogonality. 

\begin{algorithm}
  \caption{Exhaustive Search Method}
  \label{ALG_xhstvSearch}
  \begin{algorithmic}[1]
    \STATE \textbf{Input}: positive integer $k$, where $4k$ represents the order of SH-matrix
    \STATE \textbf{Output}: A set \textbf{H} of $4k$-order SH-matrix
    \STATE Construct the set of $4k$-order SH vector $V=\{\vec{v}_1,\vec{v}_2, \cdots, \vec{v}_{N_V}\}$
    \STATE Construct $L_G$, which is a list of $N_{\vec{Q}}$ length of $(4k-1)$-combination
    \STATE Initialize \textbf{H}-set to empty: $\textbf{H} \leftarrow \{\}$
    \FOR{ $n=1$ to $N_G$}
    	\STATE $\vec{Q}=\vec{o} \cup V[L_G(n)]$
    	\IF { (\textit{isHadamard}($\vec{Q}$)== 1)}
    	\STATE $\textbf{H} \leftarrow \textbf{H} \cup \vec{Q}$
		\ENDIF
	\ENDFOR
  \end{algorithmic}
\end{algorithm}
In line no.7 of the algorithm, $V[L_G(n)]$ means selecting $(4k-1)$ number of SH vectors from $V$. Note that \textbf{H} in the Algorithm \ref{ALG_xhstvSearch} refers to the set of found SH matrix after successful orthogonality test of $\vec{Q}$ by \emph{isHadamard} routine. The following examples, give an illustration on the output of the algorithm.

\begin{myexam}
With a moderate computing resource in a notebook or desktop PC, we can find all of the SH matrices for $k=1$. We generate $N_V=C(4,2)=6$ number of SH vectors and obtain $N_{\vec{Q,U}}=20$ number of unique QSH matrices. After checking all of $\vec{Q}_i$, we found that 8 of them are SH matrices. The list of the vectors given by the algorithm are shown in Fig. \ref{fig:FIG_SH4_Vectors}, which in fact is identical to the SH vectors presented in Example \ref{EXAM_4SHV}. The found SH matrices are shown in Fig.\ref{fig:FIG_SH4_Matrices}.

\begin{figure}
  \centering
  \includegraphics[scale=0.6]{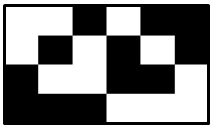}
  \caption{Entries in 4-order SH vectors: black indicates -1, white corresponds to 1 entries}
  \label{fig:FIG_SH4_Vectors}
\end{figure}

\begin{figure}
  \centering
  \includegraphics[scale=0.6]{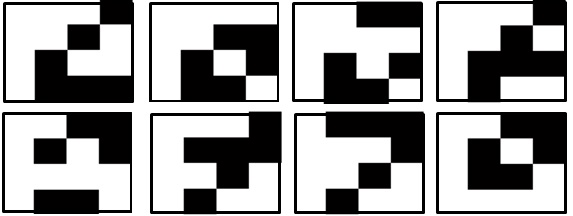}
  \caption{A list of 4-order SH matrices obtained by the exhaustive method}
  \label{fig:FIG_SH4_Matrices}
\end{figure}
 
\end{myexam}
 
\subsection{Random Vector Selection}

The RVS (Random Vector Selection) is a probabilistic algorithm which works as follows. For a $4k$-order matrix, we initially fill the first column with the unity matrix $\vec{o}$ and fill the second column with one of randomly selected $\vec{v}_1 \in V$. The third column vector $\vec{v}_2$ is also selected randomly from $V$, but if it is not orthogonal to $\vec{v}_1$, we choose another vector until we found the suitable one that is orthogonal to $\vec{v}_1$. Construction of the fourth column is conducted similarly by randomly choose a vector from $V$ and it also should pass the orthogonality test, i.e., it has to be orthogonal with the existing vectors $\{ \vec{v}_1, \vec{v}_2\}$. This process of random selection and orthogonality test are conducted up to $(4k-1)$ number of vectors are obtained, at which all of the column of the matrix are filled.

For practical consideration, such as when implementing the algorithm as a computer program, we do not generate the set $V$ because the size can be substantially large. Instead, we generate one vector at a time by random permutation of $2k$ number of $-1$ into $4k$ places, and set the remaining ones with $1$, following the definition of SH vector given in the previous section. 

The probabilistic algorithm to construct a Hadamard matrix is formulated in Algorithm \ref{ALG_RVS}. Let $4k$ be the order of the matrix we want to construct. We collect the column vector that pass the orthogonality tests in the set $H$. Output of the algorithm for $12$ order SH matrix construction is shown in Fig.\ref{fig:FIG_SH12_Matrix}.

\begin{algorithm}
  \caption{Random Vector Selection Algorithm}
  \label{ALG_RVS}
  \begin{algorithmic}[1]
    \STATE \textbf{Input}: positive integer $k$, where $4k$ represents the order of SH-matrix
    \STATE \textbf{Output}: A $4k$-order SH-matrix
    \STATE Generate the set of $4k$-order SH-vectors $V=\{\vec{v}_2, \cdots, \vec{v}_{N_V} \}$
    \STATE Fill the first column of $\vec{H}$ with $\vec{o}$: $\vec{H} \leftarrow \vec{o}$
    \STATE Fill the second column with randomly selected $\vec{v}_1 \in V$: $H \leftarrow \vec{v}_1$
    \STATE Initialize the counter: $n \leftarrow 1$
    \WHILE{$n<(4k-1)$}
    	\STATE Randomly select SH vector from $V$ and check the orthogonality to all existing column in $H$, until a vector $\vec{v}_j \in V$ that is orthogonal to already selected vectors in $\vec{H}$ is found.
    	\STATE Insert the found vector to the column of $\vec{H}$: $\vec{H} \leftarrow \vec{H} \cup \vec{v}_j$
    	\STATE $n \leftarrow n+1$
    \ENDWHILE
  \end{algorithmic}
\end{algorithm}
\begin{figure}
  \centering
  \includegraphics[scale=0.6]{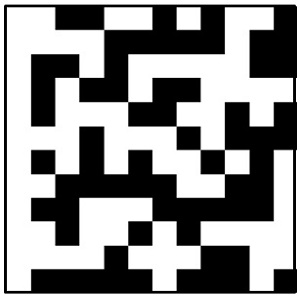}
  \includegraphics[scale=0.6]{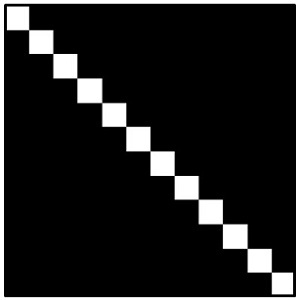}
  \caption{Output of the algorithm: left part is generated $\vec{H}$, right part is indicator matrix $\vec{H}^T\vec{H}$}
  \label{fig:FIG_SH12_Matrix}
\end{figure}

We have performed a computation experiment to construct SH matrices of orders 12, 16, 20, and 24. The iteration numbers to obtain the first 10 vectors, after setting the first one by a unity vector $\vec{o}$ and the second one with a randomly selected $\vec{v}_1 \in V$, are displayed in Table 1. In running the program, we set a maximum iteration limit where the search of the orthogonal vector is allowed, in which case the program restarted with a different random number seed.

\begin{table}
 \label{TBL_RVS_Iter}
 \centering
 \caption{Average iterations to obtain the first 10 orthogonal vectors}
 \begin{tabular}{|c|c|c|c|c|c|}
 \hline
 & & \multicolumn{4}{|c|}{Iteration Number for Corresponding Order}\\
 \hline
 No & $i$ of $\vec{v}_i$ & ${}^{12}\vec{H}$ & ${}^{16}\vec{H}$ &${}^{20}\vec{H}$ & ${}^{24}\vec{H}$ \\
 \hline
 1 & 3 & 2 & 2 & 3 & 1 \\
 2 & 4 & 5 & 6 & 9 & 15 \\
 3 & 5 & 10 & 18 & 25 & 30 \\
 4 & 6 & 58 & 82 & 60 & 25 \\
 5 & 7 & 133 & 181 & 284 & 777 \\
 6 & 8 & 118 & 342 & 499 & 2474 \\
 7 & 9 & 114 & 246 & 3160 & 4559 \\
 8 & 10 & 126 & 1003 & 8821 & 6454 \\
 9 & 11 & 296 & 1006 & 25248 & 66895 \\
 10 & 12 & 294 & 891 & 117048 & 84081 \\
 \hline
 \end{tabular}
\end{table}

Observation to each column of the table indicates that at $i^{th}$-stage, the selection of the $i^{th}$ orthogonal vector $v_i$, becomes increasingly difficult with the increasing number of already-selected vectors $ \{\vec{v}_1, \vec{v}_2, \cdots, \vec{v}_{i-1} \}$ . Additionally, the iteration number increased non-liniearly. Observation to each row indicates that the iteration number also increased in a non-linear fashion with the order of the Hadamard matrix. Since the number of iteration indicates the difficulty in finding the orthogonal $\vec{v}_i$, it means that the probability to find a vector depends on both of the stage and order of the Hadamard matrix. 

Figure \ref{fig:FIG_RVS_Iter} displays a logarithmic plot of the data listed in Table \ref{TBL_RVS_Iter}. The figure shows that curves of higher order $H$ lies above the lower ones which indicates that the higher order $H$ needs more iteration than the lower one. Additionally, the probability of finding an orthogonal vector given a selected set is decreasing with the iteration.

\begin{figure}
 \centering
 \includegraphics[scale=0.5]{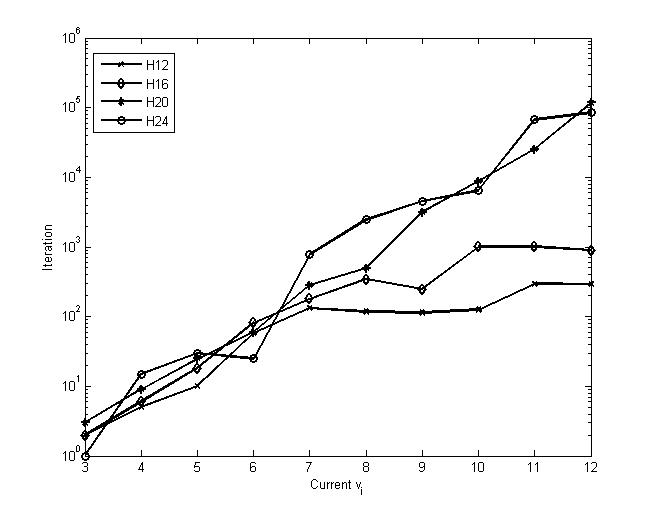}
  \caption{Semilogarithmic plot of iteration number at each stage}
 \label{fig:FIG_RVS_Iter}
\end{figure}

\subsection{Orthogonalization by Simulated Annealing}

The SA (Simulated Annealing) is a probabilistic algorithm that works by first putting a system to be optimized into a high temperature (energy) state, and then it is gradually cooled down. Transition of the system follows the first order Markov Chain, where a particular configuration $\vec{x}_n$ at stage $n$, depends only on the previous configuration $\vec{x}_{n-1}$ at stage $(n-1)$ and is not affected by the earlier system's histories at $(n-2), (n-3), (n-4), \cdots$.  

At a particular (time-) stage $n$, the system proposes a transition from $\vec{x}_n$ to $\vec{x}_{n+1}$. If it moves the system to a lower energy state, the change is accepted. However, if the system moves to a higher energy state, the proposal is only accepted with a certain amount of probability or follows a criterion. We can adopt Metropolis criterion \cite{Metropolis1953}, where the change of the system from a state with energy $E_n\equiv E(\vec{x}_n)$ to a state with energy$E_{n+1}$ is allowed with probability $min(1,e^{-(E_{n+1}-E_n)/k_BT})$, where $k_B$ is the Boltzman's constant and $T$ is the current temperature of the system. 

We also can use a scheduled cooling $T(n)$; which for simplicity is translated into threshold probability $P(T(n))$, so that a scheduled threshold probability $P(n)$, instead of scheduled temperature $T(n)$, is used. At the initial stage, the probability threshold is set to 0.5, corresponding to a high temperature stage, and then increased to probability 1.0 at the final stage.

In the SH matrix construction problem, the system $\vec{x}_n$ corresponds to the QSH matrix $\vec{Q}_n$, which is evolved from its initial state $\vec{Q}_0$ into the target state of becoming an SH matrix $\vec{Q}_N = \vec{H}$ at stage-$N$. Therefore, the evolution of the system is given by the following Markov chain:

\begin{equation}
 \vec{Q}_0 \to \vec{Q}_1 \to \cdots \to \vec{Q}_n \to \cdots \to\vec{Q}_N = \vec{H}
\end{equation}
The energy of the system $E(\vec{Q}_n)$ should be chosen to indicate the deviation of the QSH matrix from the SH matrix. Therefore, it is a measure of non-orthogonality of the $\vec{Q}_n$ which can be defined as follows

\begin{equation}
  \label{Eq1}
   E(\vec{Q}_n)=\sum_{j=1}^{4k}\sum_{i=1}^{4k} \left|D_{\vec{Q}_n}(i,j)\right|-4k \delta_{i,j} 
\end{equation}
where $D_{\vec{Q}_n}=\vec{Q}_n^T\vec{Q}_n$, and $\delta_{i,j}$ is the Kronecker delta whose value is 1 if $i=j$ and 0 otherwise. Note that the definition of energy implies that the value will be zero for an SH matrix $\vec{H}$; i.e., $E(\vec{H})=0$.

The system update is performed by randomly selecting a column; except the first one which is fixed to $\vec{o}$, and flip a pair of $-1$ and $1$ entries (i.e., we flip -1  to 1, and the entry 1 to -1), which are also randomly chosen at the selected column. The purpose of altering a pair of $\{1,-1\}$ in the column is to ensure seminormality of the column vector of $\vec{Q}_n$, i.e., it changes from  $\vec{v}_i \in V$ into another SH vector $\vec{v}_j \in V$.
After flipping the entries, the Metropolis criterion is applied, i.e., we check the system energy $E(\vec{Q}_n)$ to decide the updating proposal. If the flip into $\vec{Q}_{n+1}$ decreases the energy, we accept the update but if the flip increases the energy, we only accept with a probability, i.e., we generate a random number that is distributed uniformly in  interval $[0,1]$ and if it is greater than a particular probability threshold $P(n)$, the update is accepted; otherwise it is rejected. The detail of the algorithm and some construction examples of low-order H-matrices is described in \cite{Suksmono2016}.
%
\section{Probabilistic Analysis}

In this section, we will derive the orthogonal probability $p_{\perp}$, which is the probability that two randomly SH vectors $\vec{v}_1,\vec{v}_2 \in V$ are orthogonal. We also calculate the lower- and upper- bound estimates of $p_{\perp}$, based on permutation and combination approximation. Then, we address the distribution of H-matrices across order based on known number of unique H-matrices and compare with the results of probabilistic analysis based on $p_{\perp}$.
\subsection{Orthogonal Probability Among SH-Vectors}
Simulation in the previous sub sections indicates the relationship between the probability of finding a semi-normalized Hadamard vector  to construct a $4k$ order semi-normalized Hadamard matrix, with both of the iteration stage and the order of the matrix.  In this section, we calculate the estimation of the probability values based on basic counting principles. In deriving the results, we use the following estimates of $q$-combination of $p$-objects, $\left( \begin{array}{c} p \\ q \end{array} \right) \equiv C(p,q)$, and permutation of $p$-objects, $p!$, formulas as follows,

\begin{equation}
\label{APPROX_1}
 \left(\frac{p}{q} \right)^q \leq C(p,q) \leq \left( \frac{ep}{q} \right)^q
\end{equation}

\begin{equation}
\label{APPROX_2}
 \frac{2^{2p}}{2 \sqrt{p}} \leq C(2p,p) \leq \frac{2^{2p}}{\sqrt{2p}} 
\end{equation}

\begin{equation}
\label{APPROX_3}
  \left( \frac{p}{e} \right)^p \leq p! \leq ep \left( \frac{p}{e} \right)^p
\end{equation}

From the Corollary \ref{CORO_Num_orthov}  of the previous section, we know that there are $N_O=C(2k,k)^2$ orthogonal vectors for any $4k$ order SH vector $\vec{v}_i \in V$. An estimate of $N_O$ is given by the following lemma.

\begin{mylemm}
The estimate number of orthogonal vector $N_O$ to a $4k$-order SH vector $\vec{v}_i \in V$, is 

\begin{equation}
\frac{2^{4k}}{4k} \leq N_O \leq \frac{2^{4k}}{2k}
\nonumber
\end{equation}

\end{mylemm}

\begin{proof}
Based on the approximation formula of combination, we can derive the lower bound $N_{O,LB}$ and the upper-bound $N_{O,UB}$ estimates of the orthogonal vector as follows:

\begin{equation}
 \label{N_OLB}
 N_{O,LB} = \left( \frac{2^{2k}}{2\sqrt{k}} \right)^2=\frac{2^{4k}}{4k} 
 \nonumber
\end{equation}

\begin{equation}
 \label{N_OUB}
 N_{O,UB} = \left( \frac{2^{2k}}{\sqrt{2k}} \right)^2=\frac{2^{4k}}{2k} 
\nonumber
\end{equation}
By combining both of the bounds, we obtain the lemma. \end{proof}

We want to calculate the probability of two-randomly selected vectors $v_i, v_j \in V$ orthogonal to each other. Therefore, we also need to estimate the total number of SH vectors in $V$, which is given by the following lemma.

\begin{mylemm}
\label{LEMM_NV}
The estimate of the number of SH vectors $v_i \in V$, $N_V$, is given by $ \frac{2^{4k}}{2\sqrt{2k}} \leq N_V \leq \frac{2^{4k}}{2\sqrt{k}} $ 
\end{mylemm} 

\begin{proof}
From the Lemma \ref{LEMM_SH_number}, we know that $N_V = C(4k,2k)$. By employing the approximation formulas, we obtain the lower bound $N_{V,LB}$ and upper bound $N_{V,UB}$ respectively as follows,

\begin{equation}
 \label{N_VLB}
 N_{V,LB} = \frac{2^{4k}}{2\sqrt{2k}} 
\nonumber
\end{equation}

\begin{equation}
 \label{N_OUB}
 N_{V,UB} = \frac{2^{4k}}{2 \sqrt{2k}} 
\nonumber
\end{equation}
Combining both of the bounds, we can state the results as given in the lemma. 
\end{proof}

Since we already know both of $N_O$ and $N_V$, we can determine the value of orthogonal probability between a pair of randomly selected SH-vectors. We formulate the result into the following theorem.

\begin{mytheo}
\label{THEO_SH_ProbOrtho}
Let $\vec{v}_i, \vec{v}_j \in V$ be two randomly selected $4k$-order SH vector, with $k$ a positive  integer. The probability $p_{\perp}$ that they are orthogonal is given by $p_{\perp}=\frac{C(2k,k)^2}{C(4k,2k)-1}$
\end{mytheo}

\begin{proof}
Theorem \ref{THEO_SH_ortho} implies that a $4k$-order SH vector $\vec{v}_i \in V$ is orthogonal to (another) $N_O$-number of SH vectors $\vec{w}_j \in W \subset V$. Let $\Gamma$ be a (undirected) graph representing the orthogonality relationship among SH vectors in $V$. Then, according to Proposition \ref{PROP_SHGraph}, $\Gamma$ is an $N_O$-regular graph of order $N_V$. The orthogonal probability is therefore

\begin{equation}
p_{\perp}=\frac{Number \_ of \_ Edges\_in(\Gamma)}{Number\_of\_Edges\_in(K_{N_V})} = \frac{|\Gamma|}{|K_{N_V}|}
\nonumber
\end{equation}
where of $K_{N_V}$ is a complete graph of order $N_V$. Since $\Gamma$ is $N_O$ regular, the number of edges in is $|\Gamma|=N_V N_O /2$, whereas  the number of edges in $N_V$-order complete graph is $|K_{N_V}|=C(N_V,2)=\frac{N_V(N_V - 1)}{2}$. Therefore, the orthogonal probability is
\begin{equation}
p_{\perp}=\frac{|\Gamma|}{|K_{N_V}|}=\frac{N_O}{N_V-1} = \frac{C(2k,k)^2}{C(4k,2k)-1}
\nonumber
\end{equation}

\end{proof}
Consequently, we also obtain the following bounds of the orthogonal probability.
\begin{mycoro}
\label{CORO_Po}
Let $v_i,v_j \in V$, two randomly selected $m=4k$-order SH vectors. The estimate of probability that they are orthogonal is $\frac{1}{2 \sqrt{k}} \leq p_\perp \leq \sqrt{\frac{2}{k}}$.
\end{mycoro}

\begin{proof}

When $k \gg 1$, $C(4k,2k)-1 \simeq C(4k,2k)$, then the estimate of the lower-bound $p_{\perp,LB}$ and the upper-bound $p_{\perp,UB}$ values can be computed as follows:

\begin{equation}
p_{\perp,LB} =\frac{N_{O,LB}}{N_{V,UB}} = \frac{\frac{2^{4k}}{4k}}{\frac{2^{4k}}{2\sqrt{k}}}=\frac{2\sqrt{k}}{4k}=\frac{1}{2\sqrt{k}}
\nonumber
\end{equation}

\begin{equation}
p_{\perp,UB} =\frac{N_{O,UB}}{N_{V,LB}} = \frac{\frac{2^{4k}}{2k}}{\frac{2^{4k}}{2\sqrt{2k}}}=\frac{2\sqrt{2k}}{2k}=\sqrt{\frac{2}{k}}
\nonumber
\end{equation}
By combining both of the bounds, we will obtain the corollary.
\end{proof}
\subsection{Existence Probability, Distribution of The Hadamard Matrices, and Discrepancy Problem}

An H-matrix can be seminormalized to obtain an equivalent SH-matrix and then normalized to get an NH-matrix. We define the reverse of this process, i.e. obtaining equivalent SH-matrices from a given NH-matrix, as degeneration.

\begin{mydef}
\label{DEF_Degeneration}
Degeneration of an NH-matrix is a process of obtaining all equivalent (unique) SH-matrices from a given NH-matrix. The results are called (unique) degenerate SH-matrices.
\end{mydef}
Degeneration of an NH-matrix into unique SH-matrices can be done by combinations of column negation. Therefore, for an $m=4k$ order NH-matrix, there will be $2^{4k-1}$ ways to negate the columns, since the first column is retained as a unity vector $\vec{o}$. We call this number as the degenerate number $N_D$. Consequently, we have the following result.

\begin{mylemm}
A $4k$-order NH-matrix degenerates into $N_D=2^{4k-1}$ SH-matrices.
\end{mylemm}

The number of unique H-matrix (up to equivalence) has been studied before and described in  \cite{Hall1961}, \cite{Hall1965}, \cite{Ito1979}, \cite{Kimura1989}, \cite{Kimura1994}, \cite{Kimura1994b}, \cite{Kimura1986}, \cite{Spence1995}, and \cite{Kharagani2010}. Currently, we know that the number of unique H-matrix of orders $4$, $8$, $12$, ..., and $32$ are, subsequently, $1$, $1$, $1$, $5$, $3$, $60$, $487$, and $13,710,027$, which we define as $N_{NH}$. The number of unique H-matrix of a particular order should equal to the number of unique H-matrix (actually an NH-matrix) multiplied by the degenerate number, therefore we have the following result.

\begin{myprop}
\label{LEMM_NSH}
The number of $N_{SH}$ of unique SH matrix of order $4k$ is $N_{SH}=N_{NH}\times 2^{4k-1}$.
\end{myprop}

It is interesting to know, whether it is possible to estimate $N_{SH}$ by probabilistic analysis. To do this task, we need two quantities, i.e., the number $N_{QU}$ of (unique) QSH-matrix and the probability that the QSH-matrix is orthogonal $p_{\vec{H}|\vec{Q}}$, which will be described in the following discussions.

First, we will estimate the orthogonal probability  of the QSH-matrix based on an idea that is derived from the construction process of an SH-matrix described in Algorithm \ref{ALG_xhstvSearch}. In the algorithm, after setting the first column vector to $\vec{o}$, the algorithm randomly select the second vector  $\vec{v}_2 \in V$. We do not count the orthogonal-probability of $\vec{v}_2$ (and the next column vectors) to the existing vector in $\vec{H}$ since all of SH vectors in $V$ is orthogonal to  $\vec{o}$. In the next step, selection of the third vector is performed by randomly select a candidate vector $\vec{v}_r \in V$ and verify whether it is orthogonal to $v_2$ or not. The probability that it is orthogonal is $p_\perp$ given by Theorem \ref{THEO_SH_ProbOrtho}. After a successful test and obtain $\vec{v}_3$, our set of the selected vectors in this stage is $\{ \vec{v}_1,\vec{v}_2, \vec{v}_3\}$. The next step is obtaining the fourth vector $\vec{v}_4 \in V$ and check the orthogonality with the previously selected ones, i.e, it should pass the orthogonality test with $\{ \vec{v}_2, \vec{v}_3\}$ given by the following conditions

\begin{equation}
 (\left<\vec{v}_2,\vec{v}_3\right>=0) \wedge
 (\left<\vec{v}_3,\vec{v}_4\right>=0) \wedge
 (\left<\vec{v}_4,\vec{v}_2\right>=0)
\nonumber
\end{equation}
where $\wedge$ is the "AND" logical operator. We will denote two vectors $\vec{v}_i,\vec{v}_j $ orthogonal as $\left< \vec{v}_i, \vec{v}_j \right>=0 \equiv (\vec{v}_i \perp \vec{v}_j)$, so that we can rewrite the conditions more compactly as follows:

\begin{equation}
 (\vec{v}_2 \perp \vec{v}_3) \wedge
 (\vec{v}_3 \perp \vec{v}_4) \wedge
 (\vec{v}_4 \perp \vec{v}_2)
 \nonumber
\end{equation}

The analysis shows that, in general, the orthogonality conditions at the $i^{th}$-stage requires 2-combination of $(i-1)$ vectors. Therefore, this relationship can be represented as an $(i-1)$-order complete graph $K_{i-1}$ displayed in Fig.\ref{fig:pH_Kn}.

\begin{figure}
  \centering
  \includegraphics[scale=0.5]{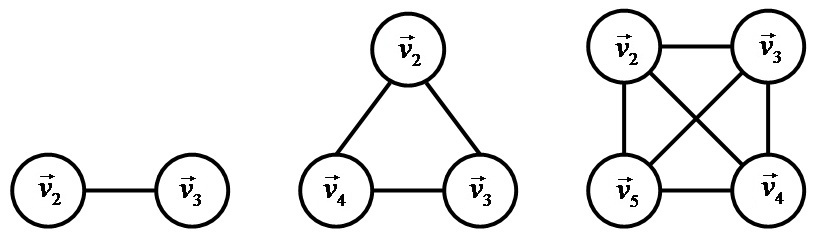}
  \caption{Orthogonality conditions that should be satisfied in each stage of the selection of a set of orthogonal vectors: two vectors: $ (\vec{v}_2 \perp \vec{v}_3)$, three vectors:  $(\vec{v}_2 \perp \vec{v}_3) \wedge
 (\vec{v}_3 \perp \vec{v}_4) \wedge
 (\vec{v}_4 \perp \vec{v}_2)$, four vectors:$ 
 (\vec{v}_2 \perp \vec{v}_3) \wedge (\vec{v}_3 \perp \vec{v}_4) \wedge (\vec{v}_4 \perp \vec{v}_5) \wedge
 (\vec{v}_5 \perp \vec{v}_2) \wedge (\vec{v}_2 \perp \vec{v}_4) \wedge (\vec{v}_3 \perp \vec{v}_5)$.}
  \label{fig:pH_Kn}
\end{figure}

In the construction of $4k$-order SH-matrix, the process of selection and orthogonality tests are conducted subsequently up to the $(4k-1)^{th}$ vector, $\vec{v}_{4k-1}$, is obtained; at this stage, the matrix is definitely an SH-matrix. The probability of this configuration is $p_\perp^{(4k-1)(4k-2)/2}$. Then, given a $4k$ order QSH-matrix, the probability that it is in fact an SH-matrix, referring to Fig.\ref{fig:pH_Kn}, can be formulated as follows.

\begin{myprop}
\label{PROP_ProbHExistence}
The probability $p_{\vec{H}|\vec{Q}}$ that a randomly generated $4k$-order QSH-matrix $\vec{Q}$ is actually an SH-matrix is 
 $p_{\vec{H}|\vec{Q}}= p_\perp^{(4k-1)(4k-2)/2}$.
\end{myprop}

\begin{proof}
Based on the complete-graph representation of orthogonality relationship, and since the number of edge $|K_{4k-1}|=(4k-1)(4k-2)/2$,  the probability that a randomly generated QSH matrix is an SH-matrix is $p_{\vec{H}|\vec{Q}}=p_{\perp}^{|K_{4k-1}|}= p_\perp^{(4k-1)(4k-2)/2}$.
 
\end{proof}

\begin{mycoro}
\label{CORO_pHQ}
  For $k \gg 1$, the estimate of $p_{\vec{H}|\vec{Q}}$ is given by $ \left(4k \right)^{-4k^2}\leq p_{\vec{H}|\vec{Q}} \leq \left( \frac{k}{2}\right)^{-4k^2}$
\end{mycoro}

\begin{proof}
The lower-bound $p_{\vec{H}|\vec{Q},LB}$ and upper- bound $p_{\vec{H}|\vec{Q},UB}$ estimate of this value are:

\begin{equation}
 p_{\vec{H}|\vec{Q},LB}=\left( 2^{-1}k^{-1/2}\right)^{(4k-1)(4k-2)/2} = 
 \left(4k \right)^{-(4k-1)(4k-2)/4} \approx \left( 4k\right)^{-4k^2}
\nonumber
\end{equation}
\begin{equation}
p_{\vec{H}|\vec{Q},UB}=\left( \left( \frac{2}{k}\right)^{1/2}\right)^{(4k-1)(4k-2)/2} = 
\left( \frac{k}{2}\right)^{-(4k-1)(4k-2)/4} = \left( \frac{k}{2}\right)^{-4k^2}
\nonumber
\end{equation}
\end{proof}

After knowing the existence probability $p_{\vec{H}|\vec{Q}}$, we need to  calculate the number $N_{QU}$ of $4k$-order QSH-matrix, which is formulated into the followings.

\begin{myprop}
\label{PROP_NQU}
  The number $N_{QU}$ of unique QSH-matrix of order $4k$ is $N_{QU}=C(N_V, 4k-1)$, whose  bounds are given by $\left( \frac{2^{4k}}{8\sqrt{2} k^{3/2}} \right)^{4k} < N_{QU} < \left( \frac{2^{4k}}{8k^{3/2}} \right)^{4k}$.
\end{myprop}

\begin{proof}
   We know that there are $N_V$ number of $4k$-order SH vectors, which is to be arranged in $4k-1$ number of column to become a QSH-matrix. Then, according to basic counting, we will have the number of QSH-matrix $N_{QU} = C(N_V,4k-1)$. The estimate of $N_{QU}$ is based on the estimate of the combination given by Eq. \ref{APPROX_2} and approximation of $N_V$ given by Lemma \ref{LEMM_NV}. Accordingly, we obtain the following bounds:

\begin{equation}
 N_{QU,LB} = \left( \frac{\frac{2^{4k}}{2 \sqrt{2k}}}{(4k-1)} \right)^{4k-1} = 
 \left( \frac{2^{4k}}{2(4k-1)\sqrt{2k}} \right)^{4k-1}
 \nonumber
\end{equation}

\begin{equation}
 N_{QU,UB} = \left( \frac{\frac{2^{4k}}{2 \sqrt{k}}}{(4k-1)} \right)^{4k-1} = 
 \left( \frac{2^{4k}}{2(4k-1)\sqrt{k}} \right)^{4k-1}
 \nonumber
\end{equation}
when $ k\gg$, we obtain $N_{QU,LB}=\left( \frac{2^{4k}}{8\sqrt{2} k^{3/2}} \right)^{4k} $ and  $N_{QU,LB}=\left( \frac{2^{4k}}{8k^{3/2}} \right)^{4k}$. 

\end{proof}
By using this result and the estimate of existence probability, we can calculate the expectation number of $4k$ order SH-matrix, $E[H] \approx p(H|Q)\times N_{QU}$, which is formulated as follows.

\begin{mylemm}
\label{LEMM_NSHPROB}
  The expectation number $E[H]$ of unique SH-matrix of order $4k$ is
  
  \begin{equation}
    \frac{2^{8k^2-14k}}{k^{4k^2+6k}} <E[H]< \frac{2^{20k^2-12k}}{k^{4k^2+6k}}
  \nonumber    
  \end{equation}
  
\end{mylemm}

\begin{proof}
  The expectation is given by $E[H] = p_{\vec{H}|\vec{Q}} \times N_{QU}$, whose lower-bound is $E[H]_{LB} = p_{\vec{H}|\vec{Q},LB} \times N_{QU,LB}$ and upper-bound is $E[H]_{UB} = p_{\vec{H}|\vec{Q},UB} \times N_{QU,UB}$. By using Corollary \ref{CORO_pHQ} and the estimate of $N_{QU}$ given by Proposition \ref{PROP_NQU}, and by assuming $k\gg$ we obtain the followings for the lower-bound  
  
  \begin{equation}
   E[H]_{LB} = \left(4k \right)^{-4k^2} \times \left( \frac{2^4k}{2(4k-1)\sqrt{2k}} \right)^{4k-1} \approx \frac{2^{-8k^2}k^{-4k^2} \times 2^{16k^2}}{2^{12k} k^{4k} 2^{2k} k^{2k}}
   \nonumber
  \end{equation}
  simplification will gives $E[H]_{LB} =\frac{2^{8k^2-14k}}{k^{4k^2+6k}} $. The upper bound is calculated as follows
  
  \begin{equation}
   E[H]_{UB} = \left(k/2 \right)^{-4k^2} \times \left( \frac{2^{4k}}{2(4k-1)\sqrt{k}} \right)^{4k-1} \approx \frac{2^{4k^2}k^{-4k^2} \times 2^{16k^2}}{2^{4k}2^{8k}k^{4k}k^{2k}}
   \nonumber
  \end{equation}
  which simplifies into $E[H]_{UB}=\frac{2^{20k^2-12k}}{k^{4k^2+6k}}$.
\end{proof}

\begin{figure}
  \centering
  \includegraphics[scale=0.5]{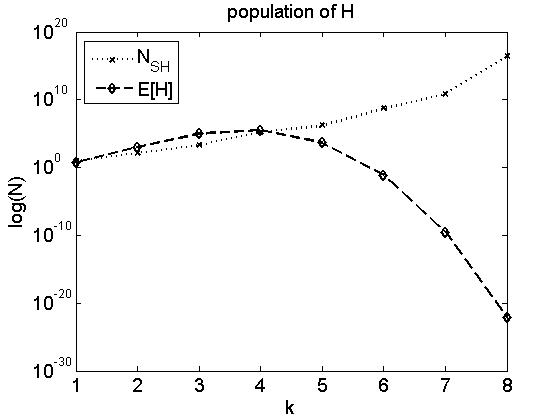}
  \caption{Comparison of $E[H]$ with $N_{SH}$ for $k=1, ..., 8$ shows the discrepancy between known number of unique SH-matrix and estimated number derived from probabilistic analysis.}
  \label{fig:EH_vs_NSH}
\end{figure}

We compare the expectation number $E[H]$ of SH-matrix given by Lemma \ref{LEMM_NSHPROB} with the number of SH-matrix $N_{SH}$ given by Proposition \ref{LEMM_NSH} for $k=1$ up to $k=8$ (matrix order 4 up to 32) by computation and plot the result in Figure \ref{fig:EH_vs_NSH}. The curves show discrepancy between $E[H]$ and $N_{SH}$ and indicate that started from $k=5$ that corresponds to order 20, the expected number of SH-matrix $E[H]$ cease to exist. Of course it contradicts with the fact that the number of the H-matrix should have been increasing with the increase of the order suggested by $N_{SH}$. This discrepancy might indicate that analyzing the orthogonality graph as a random graph $G(N_V,p_{\perp})$ is not sufficient for calculating the distribution of the SH-matrix across the orders.
%
\section{Conclusion}
We have presented a probabilistic construction method of SH-matrix, which in principle is generally applicable for any $m=4k$ order. The method build the SH matrix by selecting the column vector from a predefined SH vectors. We have also formulated  important properties of the SH-vectors, estimating the distribution of SH-matrix across the order, and found discrepancy between the expected value of the number of H-matrix obtained by the probabilistic analysis with the known value.

%

\end{document}